\documentclass[aps,pra,groupedaddress,twocolumn,noeprint]{revtex4-1}

\usepackage{graphicx}
\usepackage{dcolumn}
\usepackage{bm}
\usepackage{hyperref}
\usepackage[mathlines]{lineno}
\usepackage{amsthm,amssymb}
\usepackage{physics}
\usepackage{mathrsfs}
\usepackage{xcolor}
\usepackage[caption=false]{subfig}
\rmfamily

\definecolor{bur}{rgb}{0.5, 0.0, 0.13}

\newcommand\etal{\emph{et al.~}}

\newtheorem{lemma}{Lemma}
\DeclareMathOperator{\arcsinh}{arcsinh}

\begin{document}

\title{Necessary conditions for steerability of two qubits, from consideration of local operations}

\author{Travis J. Baker}
\email{travis.baker@griffithuni.edu.au}
\author{Howard M. Wiseman}
\affiliation{%
Centre for Quantum Dynamics, Griffith University, Brisbane, 4111 Australia \\
Centre for Quantum Computation and Communication Technology (CQC2T)
}%

\date{\today}

\begin{abstract}
EPR-steering refers to the ability of one observer to convince a distant observer that they share entanglement by making local measurements. 
Determining which states allow a demonstration of EPR-steering remains an open problem in general. 
Here, we outline and demonstrate a method of analytically constructing new classes of two-qubit states which are non-steerable by arbitrary projective measurements, from consideration of local operations performed by the steering party on states known to be non-steerable.
\end{abstract}

\maketitle

\section{\label{sec:intro}Introduction}

Steering is a remarkable feature of quantum mechanics named by Schr\"odinger in 1935 \cite{Sch35}, whereby one party (Alice) can influence the outcomes of a distant party (Bob) by performing local measurements on a shared state. 
The coining of this term was motivated by the famous EPR paper \cite{EPR35} on the nonlocality of entangled states earlier in the same year. 
More recently, EPR-steering was formalised by one of us and coworkers \cite{Wis07}.
As a consequence, it is known that the set of correlations obtainable from steering is distinct from entanglement and Bell nonlocality, as illustrated in Fig.~\ref{fig:convex_sets}.
The question of which states admit a demonstration of EPR-steering remains an open question even for the simplest case of two qubits.

\begin{figure}[htbp]
\centering
	\includegraphics[width=\linewidth]{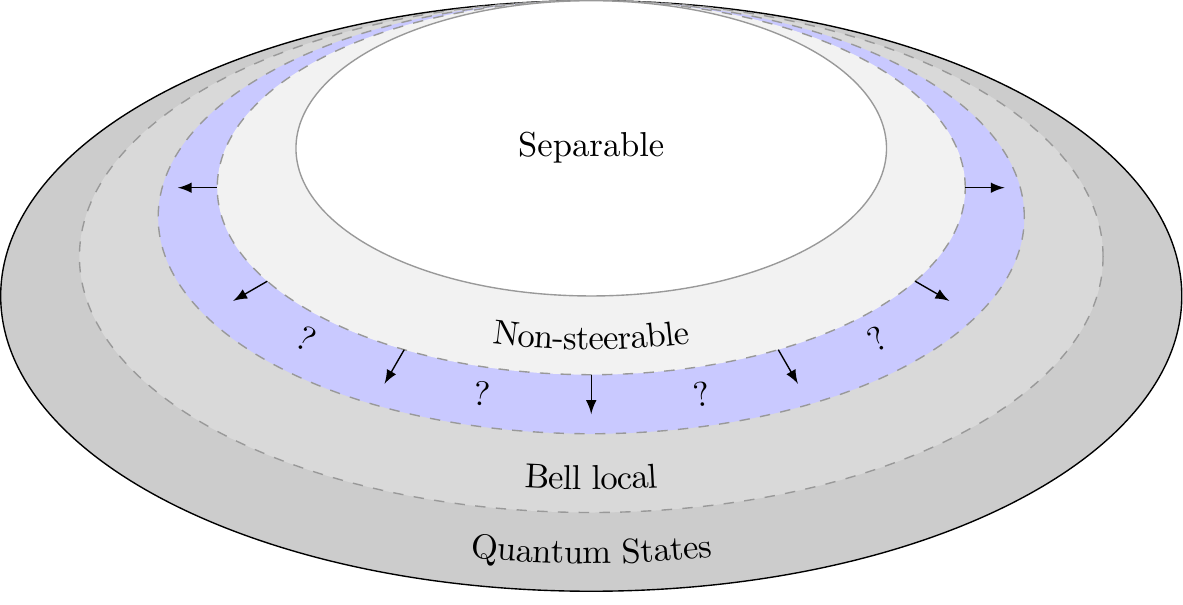}
	\caption{Heirarchy of quantum correlations. The motivation of this paper is to investigate the boundary of non-steerability for two qubit states, when the steering party can make projective measurements.}
	\label{fig:convex_sets}
\end{figure}

EPR-steering (or just ``steering'', for short) has found application in one-sided device-independent quantum key distribution \cite{Bra11}, wherein a non-zero key rate is possible only if the states used by Alice and Bob can be used to demonstrate steering.
As such, a characterization of the set of states which are able to demonstrate steering could be of great use for future quantum technologies.
Moreover, the question of proving {\em non-steerability} with respect to a given set of measurements is important in studying the phenomena of one-way steering \cite{Han12, Eva14a, Bow14, Skr14, Wol16, Tis18}, hidden steering \cite{Qui15}, and joint measurability \cite{Qui14, Uol14, Uol15}.

In this paper, we consider a steering scenario in which Alice and Bob share a pair of qubits, and are permitted to perform local projective measurements.
Although the set of two-qubit states which are steerable has not been completely characterized, partial results have been derived \cite{Wis07, Jon07, Cav09, Jev15, Ngu16a, Bow16, Cav16b, Mil18}. 
To find necessary conditions for steerability is equivalent to finding sufficient conditions for when Alice provably {\em cannot} steer Bob.
For this, one must construct a so-called Local-Hidden-State (LHS) model for Bob \cite{Wis07}; that is, show there exists a semi-classical model in which the correlations between the two parties can be described by classical post-processing of local quantum states on Bob's side.
On the other hand, a plethora of steering inequalities \cite{Cav09} have been developed, the violation of which are sufficient to demonstrate steering.

In their seminal paper, Wiseman \etal \cite{Wis07} derived necessary and sufficient conditions for the steerability of the highly symmetric two-qudit Werner state of arbitrary dimension $d$ \cite{Wer89}.
Later, Jevtic \etal \cite{Jev15} constructed LHS models for the more general class of $T$-states---two-qubit states with maximally mixed marginals---and conjectured the derived condition was precisely the border between steerable and non-steerable states.
This was later shown analytically \cite{Ngu16a}.
Beyond $T$-states, analytic necessary and sufficient conditions for two-qubit steerability are not known.
Numerical algorithms to construct necessary conditions based on semi-definite programming (SDP) techniques have been developed \cite{Hir16, Cav16a, Fil18}, where the SDP size increases exponentially in the number measurements performed by the steering party.
Moreover, an efficient numerical method to calculate (up to a desired precision) the exact border of non-steerability  for two qubits,  for infinite projective measurements made by the steering party was recently found \cite{Ngu19}. 

In this paper, we present a new method which allows the derivation of new analytic conditions for non-steerability.
Explicated here for the case of two-qubits, the idea involves allowing the untrusted party to apply local completely-positive trace-preserving (CPTP) operations, together with previous results on  $T$-state  non-steerability mentioned above, to construct two-qubit states which are provably non-steerable.
We compare the results with a  simpler  sufficient condition guaranteeing two-qubit non-steerability that was obtained by Bowles \etal \cite{Bow16}.
In their analysis, the authors noted an interesting open question; are there non-steerable states which cannot be written as convex combinations of non-steerable states satisfying their criterion, and separable states? 
The analytic results we present here answer this question in the affirmative.
We give examples to demonstrate this claim.
Moreover, we expect our method to be a useful recipe for constructing LHS models in higher dimensional systems. 

The remainder of this paper is written in the following order.
In Section \ref{sec:twoqubits} we discuss the theory of steering.
In Section \ref{sec:cptp} we briefly introduce the concept of completely-positive maps in quantum information theory, and introduce a useful characterization of such maps for qubits.
In Section \ref{sec:steeringmaps}, ideas from the previous two sections are synthesized to show how classes of two qubit states can be shown to be non-steerable based on local operations.
We give examples, before concluding in Section \ref{sec:conc}.

\section{Steering and Two-Qubit Geometry}
\label{sec:twoqubits}

Consider two distant observers Alice and Bob, who share a bipartite quantum state $\rho \in \mathfrak{B}\left( \mathcal{H}_A \otimes \mathcal{H}_B \right)$. 
Suppose Alice can perform any one of a collection of measurements $\{\mathcal{M}_{x}\}_x$ on her system, where $x$ serves as an index for each measurement. 
To be general, each such measurement is described by a Positive Operator-Valued Measure (POVM), characterized by the set of operators $\{ E_{a | x}\}_a$, with outcomes $a$, where, $\forall x$, $\sum_a E_{a | x} = I$ and $\forall a, E_{a | x}>0$.  
In a steering protocol, Bob requests that Alice make measurement $x$ on her system and announce the result $a$.
The set of Bob's reduced states, conditioned on $a$ and $x$ is called an {\em assemblage} \cite{Pus13}, which is denoted by $\{ \sigma_{a|x}\}_{a,x}$.
Specifically, the assemblage comprises the unnormalised states
\begin{equation}
	\sigma_{a|x} = \Tr_A \left[ E_{a|x} \otimes I \rho \right], \label{eq:2assemblage}
\end{equation}
which contain information regarding Alice's marginal statistics, $p(a|x) = \Tr [ \sigma_{a|x} ]$.
In a practical setting, Bob can reconstruct his assemblage by repeatedly performing tomographically complete measurements.
Bob's task is now to decide whether his assemblage could have arisen from a classical processing of local quantum states.
That is, if there exists an ensemble of {\em local-hidden states} $\{ \sigma_\lambda\}_\lambda$ described by the variable $\lambda$, a probability density $p(\lambda)$, and a non-negative response function $p(a| x,\lambda)$ such that
\begin{equation}
	\sigma_{a|x} = 
	 \int d\lambda p(a| x,\lambda) p(\lambda) \sigma_\lambda \label{eq:lhsmodel}
\end{equation}
holds $\forall a, x$.
If no such decomposition holds, it is said that Alice has succeeded in steering Bob.

In this work, we consider the subset of steering scenarios where Alice and Bob share a two-qubit state ($\dim(\mathcal{H}_A) = \dim(\mathcal{H}_B) = 2$).
In order to state our results in a useful form, we adopt the basis of Pauli operators for operators on $\mathcal{H}_{A(B)}$.
Any shared quantum state can then be expressed as
\begin{equation}
	\rho = \frac{1}{4}\left( I \otimes I + \bm{a} \cdot \bm{\sigma} \otimes I + I \otimes \bm{b} \cdot \bm{\sigma} + \sum\limits_{i,j=1}^3 T_{ij} \sigma_i \otimes \sigma_j \right),
	\label{eq:noncanonicalstate}
\end{equation}
where all information about the state is encoded into Alice's and Bob's local Bloch vectors $a_i = \Tr[\rho (\sigma_i \otimes I)]$, $b_i= \Tr[\rho (I \otimes \sigma_i)]$, and the elements of the correlation matrix $T_{ij}=\Tr[\rho (\sigma_i \otimes \sigma_j)]$.
Here, $i,j=1,2,3$.

As in \cite{Bow16}, we note that the problem of two-qubit steering can be simplified by reducing $\rho$ to a canonical form, by performing a local filtering operation on Bob's side under which steerability and non-steerability are invariant \cite{Qui15}. 
This is achieved by the transformation
\begin{equation}
	F_B(\rho) := \frac{\left( I \otimes \rho_B^{-1/2}\right) \rho \left( I \otimes \rho_B^{-1/2} \right)}{\Tr \left( I \otimes \rho_B^{-1/2}\right) \rho \left( I \otimes \rho_B^{-1/2}\right)}
\end{equation}
which leaves Bob's reduced state maximally mixed, $\bm{b}=\bm{0}$ \cite{Jev14}.
Applying this operation to Eq. \eqref{eq:noncanonicalstate} and diagonalizing $T$ by performing local unitary operations \cite{Jev15}, we have the canonical form of the two qubit state
\begin{equation}
	\rho = \frac{1}{4}\left( I \otimes I + \bm{a} \cdot \bm{\sigma} \otimes I + \sum\limits_{i,j=1}^3 t_i \sigma_i \otimes \sigma_i \right),
	\label{eq:canonicalstate}
\end{equation}
where $\{ t_i \}_{i=1}^3$ are the eigenvalues of the diagonal $T$ matrix.
Notice that the problem of determining two-qubit steerability now involves a set of only six parameters.
In this paper, we will refer to a state by its tuple $(\bm{a}, T)$.

$T$-states are the family of states which are solely determined by the matrix $T$, and therefore have $\bm{a}=\bm{0}$.
Recently, it has been shown \cite{Jev15, Ngu16} that the $T$-states which lie on the border of steerability are those which are solutions to
\begin{equation}
	\frac{1}{2\pi} \iint d^2 \bm{\hat{x}} \sqrt{\bm{\hat{x}}^\top T^2 \bm{\hat{x}}}  = 1,
	\label{eq:tstates}
\end{equation}
where the integration is over all unit vectors $\bm{\hat{x}}$ on the Bloch sphere, and $\top$ denotes matrix transpose.
Since $T$-states are equivalent to mixtures of the four Bell states, they can be geometrically represented by a tetrahedron in the $(t_1,t_2,t_3)$ space with the Bell states at the vertices \cite{Hor96}.
Eq.~\eqref{eq:tstates} implicitly defines a convex surface in this space, the closure of which consists of the entire set of non-steerable $T$-states.

Bowles \etal \cite{Bow16} showed, using a LHS model constructed with an ensemble of states which is uniformly distributed on the Bloch sphere,
that a sufficient condition for non-steerability is
\begin{equation}
	\max_{\bm{\hat{x}}} \left[ (\bm{a} \cdot \bm{\hat{x}})^2 + 2\| T\bm{\hat{x}} \| \right] \leq 1,
	\label{eq:bowles}
\end{equation}
where $\| \bullet \|$ is the Euclidean norm.
A simple calculation shows that for the Werner state ($\bm{a}=\bm{0}, T=-\mu I$), this condition is indeed tight, $\mu \leq \frac{1}{2}$ \cite{Wis07}.
However, we show in this paper that there are non-steerable states which cannot be detected by this criterion, nor by taking convex combinations between states that satisfy it and separable states.
Consequently, we are able to provide new analytic boundaries of non-steerability for certain classes of two-qubit states.

\section{Completely Positive Trace-Preserving Maps}
\label{sec:cptp}

In the context of quantum information theory, positive maps map density operators to density operators.
However, quantum systems cannot always be regarded as isolated from an environment.
Local operations must therefore ensure that any operation performed on one part of a composite system preserves the positivity of the density operator describing the entire system.
Therefore, the notion of {\em complete positivity} is required when discussing the allowable operations on quantum states.

An extremely important development in the understanding of completely positive (CP) maps was Choi's theorem \cite{Cho75}.
This theorem says that the map $\Phi: \mathbb{C}^{m\times m} \mapsto \mathbb{C}^{n\times n}$ is CP if and only if it admits a decomposition of the form
\begin{equation}
	\Phi (\rho) = \sum\limits_k A_k \rho A_k^\dagger, \label{eq:krausmap}
\end{equation}
where the $A_k$ are $n\times m$ matrices.
In order to ensure $\Phi$ is trace-preserving (TP), the constraint $\sum_k A_k^\dagger A_k = I$ is required.
The operators $\{A_k\}$ are conventionally referred to as Kraus operators, due to his influential work \cite{Kra71}.

We are principally interested in the set of CPTP maps which can be applied to single qubits.
Since a CPTP map on a single qubit is a linear operator on $\mathbb{C}^{2\times 2}$, King and Ruskai \cite{Kin99} showed that such a map can always be represented by a real $4\times 4$ matrix with six parameters.
Denoting the $2\times 2$ identity matrix by $\sigma_0$, such a map transforms the basis of Pauli matrices as $\Phi( \sigma_\nu) = \sum_\mu \mathcal{M}_{\mu\nu} \sigma_\mu$, for $\mu,\nu=0,...,3$,
where
\begin{equation}
	\mathcal{M}
		:= \begin{pmatrix}
			1 & \bm{0}^\top \\
			\bm{m} & \Lambda \\
		\end{pmatrix}
	 =
		\begin{pmatrix}
			1 & 0 & 0 & 0 \\
			m_1 & \lambda_1 & 0 & 0 \\
			m_2 & 0 & \lambda_2 & 0 \\
			m_3 & 0 & 0 & \lambda_3 \\
		\end{pmatrix}.
		\label{eq:mappingmatrix}
\end{equation}
Here, the zeros in the top row result from trace preservation, and the diagonal form of $\Lambda$ is obtained by performing a modified singular value decomposition to the part of $\mathcal{M}$ which acts on the subspace spanned by matrices with zero trace \cite{Kin99}.
As such, it was shown that one can always decompose an arbitrary CPTP map on a qubit as
\begin{equation}
	\Phi(\rho) = U \left[ \Phi_{\bm{m}, \Lambda} \left( V\rho V^\dagger\right)\right]U^\dagger,
\end{equation}
where $U,V$ are unitaries, and $\Phi_{\bm{m}, \Lambda}(\bullet)$ is the map represented by Eq.~\eqref{eq:mappingmatrix}.
Importantly, $\Phi$ will be CP iff $\Phi_{\bm{m}, \Lambda}$ is.

For reasons that will become clear in section \ref{sec:steeringmaps}, we will require a characterization of {\em extremal} maps on qubits.
An extremal map, denoted by $\Phi^\mathcal{E}$, is one that cannot be written as the convex combination of two other CPTP maps, i.e. $\Phi^\mathcal{E}(\bullet) \neq a\Phi_1(\bullet) + (1-a)\Phi_2(\bullet)$ for any $a\in(0,1)$.
The set of extremal qubit CPTP maps were studied in detail in Ref.~\cite{Rus02}, wherein the authors obtained a useful parametrization.
Crucially, it was shown that a map represented by Eq.~\eqref{eq:mappingmatrix} is extremal if and only if at most one element of $\bm{m}$ (conventionally $m_3$) is non-zero.
Moreover, the closure of the set of extreme points of CPTP maps was found to consist of those represented (up to a permutation of indices) by 
\begin{equation}
	\mathcal{M}^\mathcal{E}_{u,v} =
		\begin{pmatrix}
			1 & 0 & 0 & 0 \\
			0 & \cos u & 0 & 0 \\
			0 & 0 & \cos v & 0 \\
			\sin u\sin v & 0 & 0 & \cos u\cos v \\
		\end{pmatrix},
	\label{eq:extrememapmatrix}
\end{equation}
where each $u \in[0,2\pi), v\in [0,\pi)$ specifies an extremal map $\Phi^\mathcal{E}_{u,v}$. 
In the next section, we will apply maps of this form to the steering party's system in order to show novel results about non-steerability.

\section{Non-steerability by consideration of CPTP Maps}
\label{sec:steeringmaps}

Here, we use the following approach to tackle the problem of proving non-steerability of a two-qubit state from Alice to Bob.
Rather than explicitly constructing a LHS model for a general two-qubit state, we can begin with $T$-states which are on the border of steerability (i.e. are solutions to Eq.~\eqref{eq:tstates}), and consider the entire set of CPTP maps on Alice's qubit which give rise to non-zero $\bm{a}$.
By the following lemma, this set of states will also be non-steerable.

\begin{lemma}
Let $\rho$ be a two-qubit state initially shared between Alice and Bob, for which a LHS model exists for Bob for all projective measurements made by Alice, $\{ \Pi_{a|x} \}_a$.
If Alice performs local operations $\Phi_A$ on her qubit, the final state $(\Phi_A \otimes I)\rho$ will also admit a LHS model for Bob for all projective measurements made by Alice.
\end{lemma}
\begin{proof}
If Alice performs a CPTP map on her qubit before making her measurement, Bob's assemblage will be the collection of states
\begin{align}
	\sigma^{\Phi_A}_{a|x} &= \Tr_A \left[ (\Pi_{a|x} \otimes I) \sum_k (A_k \otimes I) \rho (A_k^\dagger \otimes I)\right] \\
	 &= \Tr_A \left[ \sum_k (A_k^\dagger \Pi_{a|x} A_k \otimes I)  \rho \right] \\
	 &\equiv \Tr_A \left[ \left( E^\Phi_{a|x} \otimes I \right) \rho \right],
\end{align}
where $ E^\Phi_{a|x} := A_k^\dagger \Pi_{a|x} A_k$ are POVM elements of a two-outcome POVM.
Now, $\rho$ admits a LHS model from Alice to Bob for all projective measurements by assumption.
Since it is known that steering with projective measurements is equivalent to steering with two-outcome POVMs (see e.g. Lemma 3 of \cite{Ngu16}), it follows trivially that $\rho$ also admits a LHS model for the subset of two-outcome POVMs $\{ E^\Phi_{a|x}\}_{a}$.

\qedhere
\end{proof}

Since we are interested in proving non-steerability for states in a general form $(\bm{a}, T)$, we need to know how a two qubit state \eqref{eq:noncanonicalstate} changes under a map applied by Alice.
We have that
\begin{align}
	(\Phi_A \otimes I)\rho &= \frac{1}{4}\Big( I \otimes I + (\bm{m} + \Lambda\bm{a}) \cdot \bm{\sigma} \otimes I + I \otimes \bm{b} \cdot \bm{\sigma} \nonumber\\
	&\qquad+ \sum\limits_{i,j=1}^3 (\Lambda_{ii}T_{ij} + m_i b_j )\sigma_i \otimes \sigma_j \Big).
\end{align}
Thus, we see that action of $\Phi_{\bm{m}, \Lambda}$ on $\rho$ induces affine transformations of the original state objects $\bm{a}, \bm{b}$ and $T$ according to
\begin{align}
	\bm{a}' &= \bm{m} + \Lambda\bm{a} \label{eq:map1}\\
	\bm{b}' & = \bm{b} \label{eq:map2} \\ 
	T' &= \bm{m} \bm{b}^\top + \Lambda T, \label{eq:map3}
\end{align}
where primes denote the object after the map is applied.
Considering the canonical form of $\rho$ (Eq.~\eqref{eq:canonicalstate}), only the relations $\bm{a}'=\bm{m} + \Lambda \bm{a}$ and $T' = \Lambda T$ are relevant, with $T$ diagonal.

We give three examples of this idea below, all which have a Bloch vector on Alice's side in the $z$-direction, for simplicity.
That is, by starting with a $T$-state which satisfies Eq.~\eqref{eq:tstates}, we construct a boundary set of non-steerable states $(a \bm{\hat{e}_z}, T')$, which satisfy
\begin{equation}
	\frac{1}{2\pi} \iint d^2 \bm{\hat{x}} \sqrt{\bm{\hat{x}}^\top (\Lambda^{-1}_{u,v} T')^2 \bm{\hat{x}}}  = 1,
	\label{eq:mappedtstates}
\end{equation}
with the constraint that $a=\sin u\sin v$. 
Here, $\bm{\hat{e}_z}$ is the unit vector in the $z$-direction.

Since the ``mapped'' surface will be a convex set, only the set of extremal maps is important, since the image of non-extremal maps will consist of states in the interior of this surface.
In other words, we consider the maps represented by Eq.~\eqref{eq:extrememapmatrix} without loss of generality.
Note that since the new surface defines a convex set of non-steerable states, the equation defining its exterior is a \emph{necessary condition} for steerability of two-qubit states.

The following classes of states we will analyze are equivalent to the so-called two-qubit {\em X-states}.
In the form of a canonical state (Eq.~\eqref{eq:canonicalstate}), X-states are those for which $\bm{a}$ is pointing only in the $z$-direction. 
That is, $\bm{a} = a \hat{\bm{e}}_z$. 
This class of states have received considerable attention in the literature due to exhibiting interesting quantum correlations \cite{Ali10, Men14}.

\subsection{Displaced Werner States}

First, we consider the problem of constructing LHS models for the family of two-qubit states which have correlation matrix proportional to the identity operator, $T'=-\xi I$, and $\bm{a}\neq \bm{0}$.
This family is similar to the family of Werner states \cite{Wer89}, without the restriction that Alice's Bloch vector vanishes.
For this reason, we refer to ${\bm a}$ as being \textit{displaced} from the centre of the Bloch sphere.
Note due to the symmetry on $T$, only the magnitude of $\bm{a}$ is important, $a:=|\bm{a}|$.  
Therefore, this family is completely characterized by two parameters, $a$ and $\xi$.

\begin{figure}[htbp]
\centering
	\includegraphics[width=\linewidth]{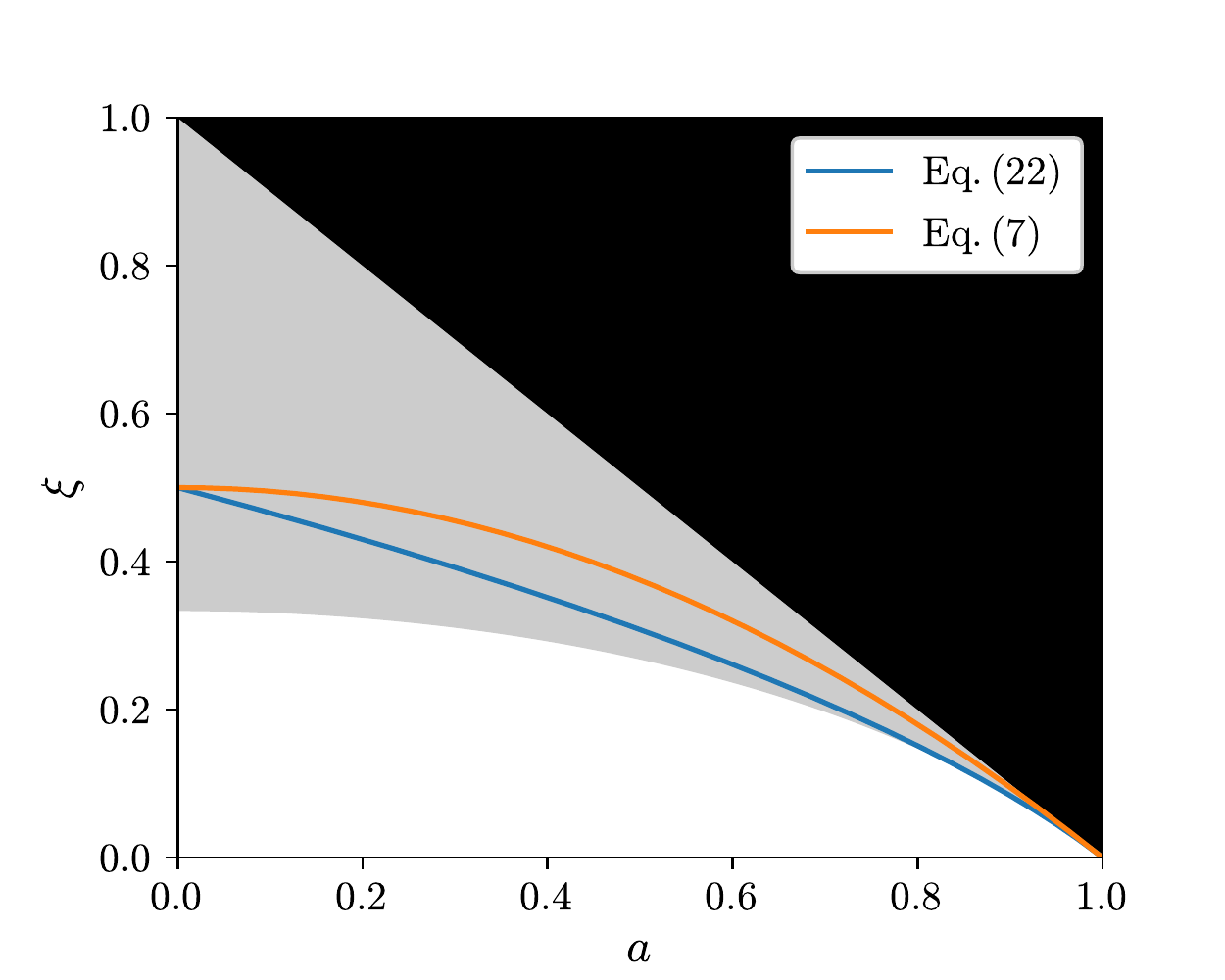}
	\caption{The sections of the state-space of the displaced Werner state for which LHS models are known.
	White region contains separable states, grey region entangled states, and those in the black region are not quantum states.
	Non-steerable regions are below the orange (Eq.~\eqref{eq:bowles}) and blue lines (Eq.~\eqref{eq:dwmaps}).
	The border constructed by performing CPTP maps on Alice's side of non-steerable T-states does not improve upon known results here.}
	\label{fig:modified_werner}
\end{figure}

There are infinitely many ways to achieve a transformation from $T$-states to the family $(a, -\xi I)$ of states, by varying $u,v$.
From Eq.~\eqref{eq:mappedtstates}, values of the correlation parameter $\xi$ on the new boundary of non-steerability will satisfy
\begin{equation}
	\xi \iint d^2\bm{\hat{x}} \sqrt{ \frac{1}{\cos^2{u}}x_1^2 +  \frac{1}{\cos^2{v}}x_2^2 + \frac{1}{\cos^2{u}\cos^2{v}}x_3^2 } = 2\pi \label{eq:xiintegral}
\end{equation}
for a given fixed value of
\begin{equation}
	|a|=\sin u\sin v.
\end{equation}
Numerical optimization shows the maximum value of $\xi$ (and hence the minimum value of the double integral in Eq.~\eqref{eq:xiintegral}) is obtained for the choice $u=v$ for all fixed values of $a$.
From this, it is straightforward to show that the resulting class of mapped displaced Werner states which are non-steerable satisfy
\begin{equation}
	\xi \leq (1-a) \left[ \frac{1}{\sqrt{1-a}} + \sqrt{\frac{1-a}{a}} \arcsinh \sqrt{\frac{a}{1-a}} \right]^{-1}.
	\label{eq:dwmaps}
\end{equation}
This condition, when equality is achieved, is plotted as the blue line in Fig.~\ref{fig:modified_werner}.
This is compared with Eq.~\eqref{eq:bowles}---all states below these lines are non-steerable.
Also shown is the set of entangled (grey), separable (white) and non-quantum (black) states.
It is evident that our method does not improve the known border of non-steerability in this case.
In the next two examples, we break this symmetry to demonstrate the usefulness of our method.

\subsection{Displaced $T$-states\label{sec:2dof}}

Next we show an example where this method improves upon the condition for non-steerability in Eq.~\eqref{eq:bowles}. 
Consider the {\em displaced} $T$-state.
That is, the most general two-qubit state in the canonical form \eqref{eq:canonicalstate}, characterized by a diagonal correlation matrix $T'=\text{diag}[t_1',t_2',t_3']$ and a {\em non-zero} Bloch vector on Alice's side, $\bm{a}\neq \bm{0}$.
Note the direction of $\bm{a}$ is important here; for simplicity, we set $\bm{a}=a \bm{\hat{e}_z}$. 

The surface of non-steerable states in which we are interested is convex in the three dimensional space parametrized by $(t_1',t_2',t_3')$.
Since the correlations---in other words, the non-zero values of $T$---are reduced in magnitude when Alice ``creates'' a non-zero $\bm{a}$ through the map $\Phi^\mathcal{E}_{u,v}$, this surface will lie inside that defined by Eq.~\eqref{eq:tstates}.
We write a vector in this space in the form
\begin{align}
	\begin{pmatrix}t_1' \\ t_2' \\ t_3'\end{pmatrix}= \begin{pmatrix}r\sin\alpha\cos\beta \\ r\sin\alpha\sin\beta \\ r\cos\alpha \end{pmatrix},
	\label{eq:sphericalcoords}
\end{align}
for $\alpha\in[0,\pi], \beta\in[0,2\pi]$.
This allows an efficient calculation of the non-steerable surface for a given $a$, by fixing $\alpha,\beta$ and maximizing $r$, the distance from the origin by computing
\begin{widetext}
\begin{equation}
	r = 2\pi \left[ \min\limits_{u,v} \iint d^2\bm{\hat{x}} \sqrt{ \frac{\sin^2\alpha\cos^2\beta}{\cos^2{u}}x_1^2 +  \frac{\sin^2\alpha\sin^2\beta}{\cos^2{v}}x_2^2 + \frac{\cos^2\alpha}{\cos^2{u}\cos^2{v}}x_3^2 } \,\right]^{-1}.\label{eq:tmapintegral}
\end{equation}
\end{widetext}
Note this minimization is effectively over one variable due to the constraint $a=\sin u\sin v$.
The following two examples show how this method can uncover new regions of non-steerable two qubit states, obtained from conditions involving the state parameters. 
These are illustrated by the blue lines in Fig.~\ref{fig:tfigs}.

We focus on the octant of the $T'$ space for which $t'_i \leq 0 ~\forall~ i$.
First, we analyse the class of displaced $T$-states that have a degenerate eigenvalue of $T$, $t_1'=t_2'$.
That is, the class of $(a \bm{\hat{e}_z}, \text{diag}[t_1',t_1',t_3'])$ states. 
As above, numerical simulations indicate the minimum in Eq.~\eqref{eq:tmapintegral} for all fixed values of $a$ is obtained for $u=v$, as a consequence of the rotational symmetry of the mapped state about the $z$-axis.
From this, using known integrals \cite{Gra14}, we can show that the surface of non-steerability is defined in terms of Eq.~\eqref{eq:sphericalcoords} by
\begin{equation}
	r = \frac{\sqrt{2(1-a)}}{\sin \alpha}\cdot\begin{cases} 
      \left[  \sqrt{1+c} + \frac{1}{\sqrt{c}} \arcsin\sqrt{-c}\right]^{-1}, & c < 0 \\
      \left[ \sqrt{1+c} + \frac{1}{2\sqrt{c}} \ln \frac{\sqrt{1+c} + \sqrt{c}}{\sqrt{1+c} - \sqrt{c}} \right]^{-1}, & c > 0
   \end{cases}
	\label{eq:2tsol}
\end{equation}
where $c := 2\cot \alpha/(1-a) - 1$.
In Fig.~\ref{fig:2DT}, we show this surface for the exemplary case $a = 0.1$.  
Non-steerable states lie to the top-right of this curve in Fig.~\ref{fig:2DT}.
This is compared to Eq.~\eqref{eq:bowles}, which guarantees the states lying to the top-right of the orange line are non-steerable.
Away from the regions where $t_1'=t_2' \approx t_3'$, our method detects new regions of non-steerable states, which are nevertheless entangled. 
Moreover, an even larger class of non-steerable states can be found by forming the convex hull of both non-steerable regions.

In Fig.~\ref{fig:3DT}, we show the same octant of the $T'$ state space for a state with parameters $t_1'=2t_2'$ and $a=0.2$.  
The surface is computed numerically by performing the optimization in Eq.~\eqref{eq:tmapintegral} by varying $\alpha \in [\pi/2,\pi]$.
By breaking the degeneracy in the eigenvalues of $T$, we find the border of non-steerability generated by CPTP maps lies entirely outside that given by Eq.~\eqref{eq:bowles} in this case.
That is, our method provides a boundary,  upon which every state in the class considered gives a positive answer to the open question noted in Ref.~\cite{Bow16}. 

Generally speaking, we have found the method presented here  is most useful  for states where the correlations between the two parties are asymmetric over the Bloch sphere, and the marginal statistics of the steering party are relatively close to maximally mixed ({\em i.e.} $a$ is small).
Finally, we note that our method can be extended to general two-qubit states; those for which the Bloch vector $\bm{a}$ is not in the $z$-direction.
This could be achieved by considering compositions of up to three extremal CPTP maps represented by Eq.~\eqref{eq:extrememapmatrix}, allowing for the permutation of indices.
That is, the composition $\Phi_3^\mathcal{E}(\Phi_2^\mathcal{E}(\Phi_1^\mathcal{E}(\bullet)))$ could be used to map a $T$-state into a nonzero-measure subset of the set of two-qubit states with $\bm{b}=\bm{0}$. 

\begin{figure}[htp]

\subfloat{%
  \includegraphics[clip,width=\columnwidth]{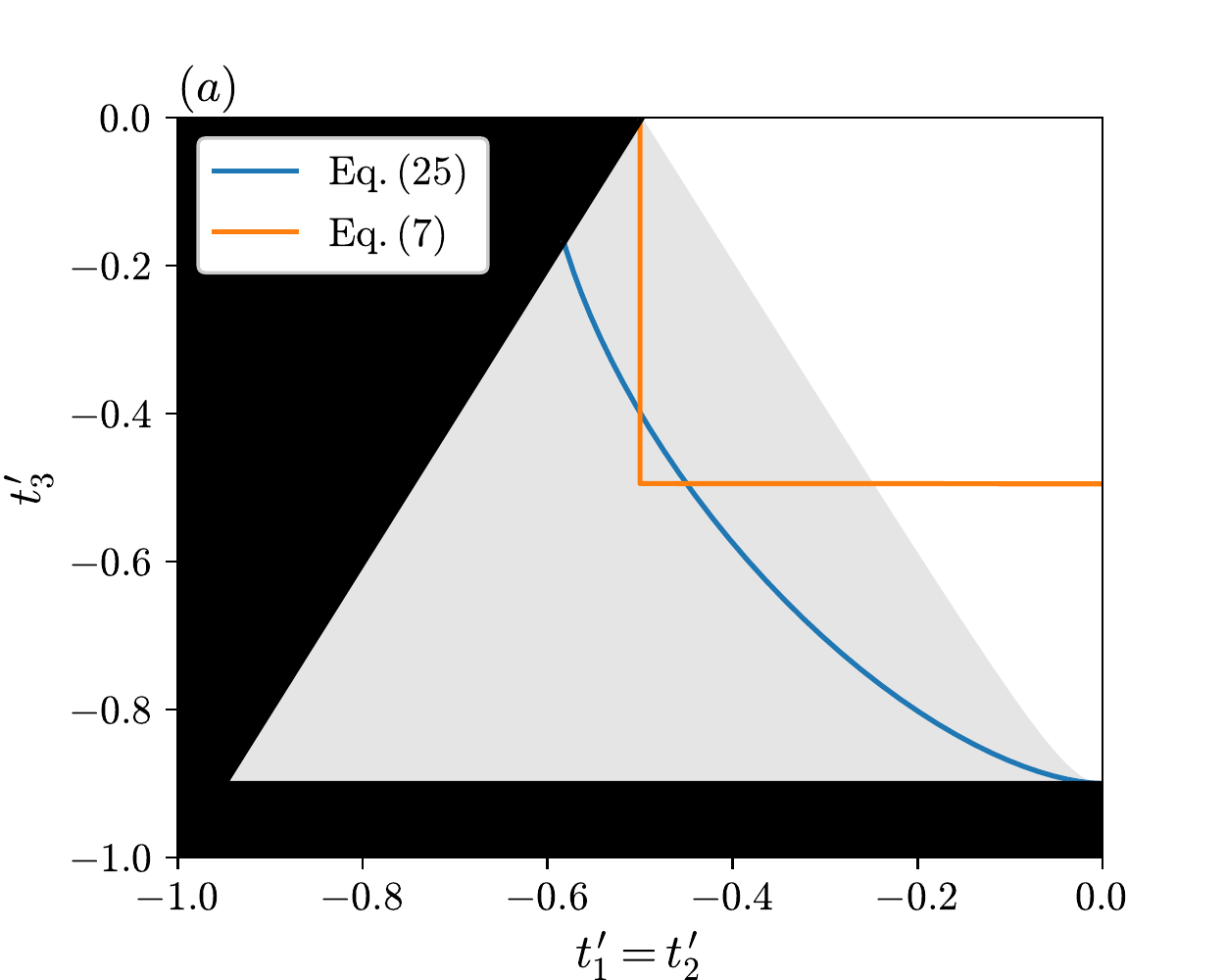}%
  \label{fig:2DT}
}

\subfloat{%
  \includegraphics[clip,width=\columnwidth]{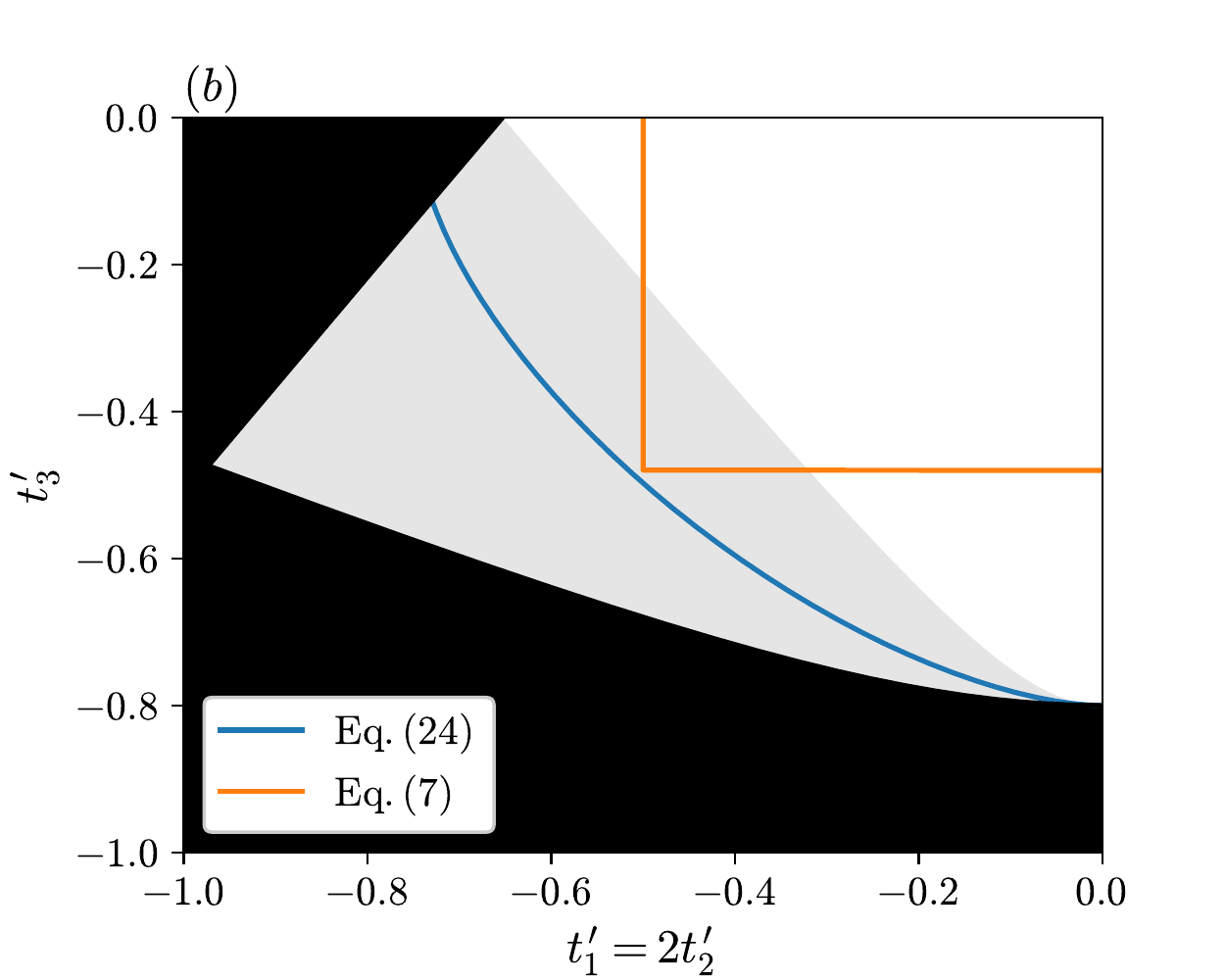}%
  \label{fig:3DT}
}
\caption{Borders of regions guaranteed to be non-steerable for two types of displaced $T$-states in the section of the $T'=\text{diag}[t_1',t_2',t_3']$ space where all $t_i'$s are negative. The white region denotes separable states, the grey region entangled states, and the black region non-quantum states.
Orange curves coincide with equality in Eq.~\eqref{eq:bowles}, blue curves are the convex hull of states obtained by applying extremal maps on Alice's part of a T-state satisfying Eq.~\eqref{eq:tstates}.
States to the top-right of these curves are non-steerable.
(a) Non-steerability borders for two qubit states where $\bm{a}=0.1\bm{\hat{e}_z}$ and $t_1'=t_2'$, given by Eq.~\eqref{eq:2tsol}. 
The methods presented in this paper uncover interesting regions of non-steerable states, away from $t_1' =t_2' \approx t_3'$.
(b) Non-steerability borders for  $(\bm{a}=0.2\bm{\hat{e}_z}, T'=\text{diag}[t_1',t_2'/2,t_3'])$ states.
By applying maps on Alice's qubit prior to steering, we can prove a set of non-steerable states which is a superset of those guaranteed by Eq.~\eqref{eq:bowles}.
}
\label{fig:tfigs}
\end{figure}

\section{Conclusion\label{sec:conc}}

We have shown a new method for proving two-qubit states are non-steerable.
These imply necessary conditions for steerability. 
Beginning with the class of T-states which lie on the border of non-steerability, we have shown how applying a CPTP map to Alice's qubit can uncover  previously unexplored  regions of non-steerable states.
These conditions may provide useful in studying phenomena where proving non-steerability is crucial, such as one-way steering and joint-measurability.
It would be interesting to see if the method presented here can be extended to steering with higher dimensional systems, with both finite and arbitrary projective measurements,  and even POVMs. 

\begin{acknowledgments}

 We would like to thank H.~Chau Nguyen, Otfried G\"uhne and Huy-Viet Nguyen for helpful comments on the preprint,  and Michael Hall for helpful comments on early results included in this manuscript.
T.J.B. acknowledges financial support through an Australian Government Research Training Program Scholarship.
This research was supported by the Australia Research Council Centre of Excellence CE170100012.
\end{acknowledgments}

\bibliography{steering_cptp_maps.bib}{}
\bibliographystyle{apsrev4-1}

\appendix

\end{document}